\documentclass[a4paper,runningheads]{llncs}

\RequirePackage[latin1]{inputenc}
\usepackage[english]{babel}

\usepackage{times}
\usepackage{helvet}
\usepackage{courier}
\usepackage{graphicx}
\usepackage{latexsym}
\usepackage{amssymb}
\usepackage{amsmath}
\usepackage{enumitem}
\usepackage{framed}
\usepackage[colorinlistoftodos, textwidth=3cm]{todonotes} 
\usepackage[spaces,hyphens]{url}
\usepackage{xspace}
\usepackage{array}
\usepackage{listings}
\usepackage{hyperref}

\title{A Note on Reasoning on $\DLliteR$ with Defeasibility}
\subtitle{Technical Report}

\author{
	Loris Bozzato\inst{1} \and Thomas Eiter\inst{2} \and Luciano Serafini\inst{1}
}

\institute{
  Fondazione Bruno Kessler, 
  Via Sommarive 18, 38123 Trento, Italy \\
  \and
  Institute of Logic and Computation, Technische Universit\"{a}t Wien,\\
  Favoritenstra\ss e 9-11, A-1040 Vienna, Austria\\
  \medskip
 \email{\texttt{\{bozzato,serafini\}@fbk.eu}, eiter@kr.tuwien.ac.at}
}

\def\isa{\sqsubseteq}

\def\I{\mathcal{I}}

\newcommand{\Pair}[2]{\left\langle#1,#2\right\rangle}

\newcommand{\vc}[1]{\mathbf{#1}}

\newcommand{\ee}{{\vc{e}}}

\newcommand{\NI}{\mathrm{NI}}
\newcommand{\NR}{\mathrm{NR}}
\newcommand{\NC}{\mathrm{NC}}

\newcommand{\IC}{\mathfrak{I}}

\newcommand{\K}{\mathcal{K}}

\newcommand{\dlmodels}{\mathop{\models_{\mathrm{DL}}}}
\newcommand{\T}{\mathcal{T}}

\newcommand{\stru}[1]{\langle #1 \rangle}

\newcommand{\non}{\neg}

\newcommand{\subs}{\sqsubseteq}

\newcommand{\Acal}{{\cal A}}

\newcommand{\Ecal}{{\cal E}}

\newcommand{\Kcal}{{\cal K}}
\newcommand{\Ical}{{\cal I}}
\newcommand{\Lcal}{{\cal L}}

\newcommand{\Rcal}{{\cal R}}
\newcommand{\Tcal}{{\cal T}}

\newcommand{\cov}[1]{\preceq}

\newcommand{\CAS}{\mathit{CAS}}

\newcommand{\OVR}{\mathit{OVR}}
\newcommand{\casmap}{\chi}
\newcommand{\mi}[1]{\mathit{#1}}

\newcommand{\nop}[1]{}

\newcommand{\SROIQrl}{\mathcal{SROIQ}\text{-RL}}

\newcommand{\NCs}{\NC_\Sigma}

\newcommand{\NIs}{\NI_\Sigma}

\newcommand{\NRs}{\NR_\Sigma}

\newcommand{\DC}{\mathbf{C}}
\newcommand{\DP}{\mathbf{P}}
\newcommand{\DV}{\mathbf{V}}

\newcommand{\subClass}{{\tt subClass}}

\newcommand{\subEx}{{\tt subEx}}
\newcommand{\subRole}{{\tt subRole}}

\newcommand{\supEx}{{\tt supEx}}

\newcommand{\pDis}{{\tt dis}}
\newcommand{\pInv}{{\tt inv}}
\newcommand{\pIrr}{{\tt irr}}

\newcommand{\nom}{{\tt nom}}
\newcommand{\cls}{{\tt cls}}
\newcommand{\rol}{{\tt rol}}
\newcommand{\const}{{\tt const}}

\newcommand{\triple}{{\tt triple}}

\newcommand{\supNot}{{\tt supNot}}

\newcommand{\default}{{\mathrm D}}

\newcommand{\ovr}{{\tt ovr}}
\newcommand{\naf}{\mathop{\tt not}}
\newcommand{\grd}{\mathit{grnd}}
\newcommand{\Head}{\mathit{Head}}
\newcommand{\Body}{\mathit{Body}}

\newcommand{\rif}{\leftarrow}

\newcommand{\insta}{{\tt insta}}
\newcommand{\instd}{{\tt instd}}
\newcommand{\triplea}{{\tt triplea}}
\newcommand{\tripled}{{\tt tripled}}

\newcommand{\ninsta}{{\tt ninsta}}
\newcommand{\ntriplea}{{\tt ntriplea}}

\newcommand{\AllNRel}{{\tt all\_nrel}}
\newcommand{\AllNRelStep}{{\tt all\_nrel\_step}}
\newcommand{\first}{{\tt first}}
\newcommand{\nextp}{{\tt next}}
\newcommand{\lastp}{{\tt last}}

\newcommand{\definst}{{\tt def\_insta}}
\newcommand{\deftriple}{{\tt def\_triplea}}
\newcommand{\defninst}{{\tt def\_ninsta}}
\newcommand{\defntriple}{{\tt def\_ntriplea}}
\newcommand{\defsubs}{{\tt def\_subclass}}

\newcommand{\defsubex}{{\tt def\_subex}}
\newcommand{\defsupex}{{\tt def\_supex}}

\newcommand{\defsubr}{{\tt def\_subr}}

\newcommand{\defdis}{{\tt def\_dis}}
\newcommand{\definv}{{\tt def\_inv}}
\newcommand{\defirr}{{\tt def\_irr}}

\newcommand{\defsupnot}{{\tt def\_supnot}}

\newcommand{\DKB}{\mathcal{K}}

\newcommand{\conp}{\ensuremath{\mathrm{co\mbox{-}NP}}}
\newcommand{\nlogspace}{\ensuremath{\mathrm{NLogSpace}}}

\def\qed{$\Box$}
\def\endproof{\ifhmode\nobreak\qed\par\fi\medskip}

\setitemize{label=--,leftmargin=*}
\setenumerate{leftmargin=*}
\newcommand{\EndEx}{\mbox{}~\hfill$\Diamond$} 

\newcommand{\comment}[1]{{#1}}

\newcommand{\lbnote}[1]{}
\newcommand{\lbreview}[1]{}

\newcommand{\DLliteR}{\textsl{DL-Lite}_\Rcal}
\newcommand{\DLlite}{\textsl{DL-Lite}}

\begin{document}

\maketitle
\vspace*{-1\baselineskip}

\begin{abstract}
  Representation of defeasible information
  is of interest in description logics, 
  as it is related to the need of accommodating exceptional instances in knowledge bases.
  %
  In this direction, in our previous works we presented
  a datalog translation for reasoning on
  (contextualized) OWL RL knowledge bases with a notion of
  justified exceptions on defeasible axioms.
  While it
  covers a relevant fragment of OWL,
  the resulting reasoning process 
  needs a complex encoding 
  in order to capture reasoning on negative information.
	%
  In this paper, 
  we consider the case of knowledge bases in $\DLliteR$, i.e. the language underlying OWL QL.
	We provide a definition for $\DLliteR$ knowledge bases with defeasible axioms and
	study their properties. 
	The limited form of $\DLliteR$ axioms allows us to formulate a simpler 
	encoding into datalog (under answer set semantics)
	with direct rules for reasoning on negative information.
	The resulting materialization method gives rise to a complete
	reasoning procedure for instance checking in $\DLliteR$ 
	with defeasible axioms.
	
\end{abstract}


\section{Introduction}
\label{sec:intro}

Representing defeasible information
is a topic of interest in the area of description logics (DLs), 
as it is related to the need of accommodating the presence of 
exceptional instances in knowledge bases.
This interest led to different proposals 
for 
non-monotonic features in DLs
based on different 
notions of defeasibility, e.g.~\cite{BonattiFPS:15,BonattiLW:06,DBLP:conf/jelia/BritzV16,GiordanoGOP:13}.
%
In this direction, we presented in~\cite{BozzatoES:18} 
an approach to represent
defeasible information in
contextualized DL knowledge bases
by introducing a notion of \emph{justifiable exceptions}:
general \emph{defeasible axioms} can be overridden
by more specific exceptional instances if 
their application would provably lead to inconsistency.
Reasoning in $\SROIQrl$ (i.e. OWL RL) knowledge bases 
is realized by
a translation to datalog, 
which provides a complete \emph{materialization calculus}~\cite{Krotzsch:10} 
for instance checking and conjunctive query (CQ) answering.
%
While the translation 
covers the full $\SROIQrl$ language,
it needs a complex encoding 
to represent reasoning on exceptions.
%
In particular, 
it relies on the use of proofs by contradiction 
to ensure completeness in presence of negative disjunctive information.




In this paper, we consider the case of knowledge bases with defeasible axioms
in $\DLliteR$~\cite{CalvaneseGLLR07}, which 
corresponds to the language underlying the OWL QL fragment~\cite{Motik:09:OWO}.
It is indeed interesting to show the applicability of our 
defeasible reasoning approach to the well-known $\DLlite$ family: 
in particular, by adopting $\DLliteR$ as the base logic
we need to take unnamed individuals introduced by existential formulas into account, 
especially for the justifications of exceptions.
Moreover, we show that
due to the restricted form of its axioms,
the $\DLliteR$ language allows us to give a less involved datalog encoding 
in which reasoning on negative information 
is directly encoded 
in datalog rules (cf.\ discussion on ``justification safeness'' in~\cite{BozzatoES:18}).

\smallskip\noindent
The contributions of this paper can be summarized as follows:
\begin{itemize}
\item
  In Section~\ref{sec:dkb} we provide a definition of defeasible DL
  knowledge base (DKB) with justified models
  that draws from the definition of 
	\emph{Contextualized Knowledge Repositories (CKR)}
	\cite{BozHomSer:DL2012,BozzatoSerafini:13,serafini-homola-ckr-jws-2012}
	with defeasible axioms provided in~\cite{BozzatoES:18}. 
	This allows us to concentrate on the defeasible reasoning aspects 
	without considering the aspects related to the representation of context
	in the CKR framework.
\item 
For DKBs based on $\DLliteR$,
we  provide in Section~\ref{sec:translation} a 
translation to datalog (under answer set semantics~\cite{gelf-lifs-91})
that alters the CKR translation in~\cite{BozzatoES:14,BozzatoES:18} 
 and prove its correctness
with respect to instance checking.
	In particular, the fact that reasoning on 
	negative disjunctive information is not needed
	allow us to provide a simpler 
	translation (without the use of the involving ``test'' environments 
	mechanism of~\cite{BozzatoES:18}).
\item
	In Section~\ref{sec:complexity}  
	we provide complexity results for 
	reasoning problems on $\DLliteR$-based DKBs.
        Deciding satisfiability of such a DKB with respect to justified models is
        tractable, while inference of an axiom under cautious (i.e.,
        certainty) semantics is \conp-complete in general.
\end{itemize}


\section{Preliminaries}
\label{sec:prelims}
\noindent
\textbf{Description Logics and $\DLliteR$ language.}
We assume the common definitions of description logics~\cite{dlhb}
and the definition of the logic $\DLliteR$~\cite{CalvaneseGLLR07}:
we summarize in the following the basic definitions used in this work.

A \emph{DL vocabulary} 
$\Sigma$ consists of the mutually disjoint countably infinite
sets $\NC$ of \emph{atomic concepts},
$\NR$ of \emph{atomic roles}, and 
$\NI$ of \emph{individual constants}.
%
Complex \emph{concepts} are then recursively defined as the smallest
sets containing all concepts that can be inductively constructed using
the constructors of the considered DL language.
%
A $\DLliteR$ \emph{knowledge base} $\K=\stru{\T,\Rcal,\Acal}$ consists of: 
a TBox $\T$ containing \emph{general concept inclusion (GCI)} axioms $C \subs D$ 
where $C,D$ are concepts, of the form:
\begin{align}
  C & := A \;|\; 
         \exists R
\\
  D & := A \;|\; 
         \non C \;|\;
         \exists R
     \end{align}
where $A \in \NC$ and $R \in \NR$;
an RBox $\Rcal$ containing \emph{role inclusion (RIA)} axioms $S \subs R$, 
reflexivity, irreflexivity, inverse and
role disjointness axioms, where $S,R$ are roles; 
and an ABox $\Acal$ composed of assertions of the forms 
$D(a)$, where $D$ is a right-side concept, $R(a,b)$, 
with $R \in \NR$ and $a,b \in \NI$.

A \emph{DL interpretation} is a pair $\I=\stru{\Delta^\I,\cdot^\I}$ where $\Delta^\I$
is a non-empty set called \emph{domain} and $\cdot^\I$ is the \emph{interpretation
function} which assigns denotations for language elements:
$a^\I \in \Delta^\I$, for $a \in \NI$;
$A^\I \subseteq \Delta^\I$, for $A \in \NC$; 
$R^\I \subseteq \Delta^\I\times\Delta^\I$, for $R \in \NR$. 
The interpretation of non-atomic concepts and roles is defined by the evaluation 
of their description logic operators (see~\cite{CalvaneseGLLR07} for $\DLliteR$).
%
An interpretation $\I$ \emph{satisfies} an axiom 
$\phi$, denoted
$\I\dlmodels\phi$, if it verifies the respective semantic condition, in particular: 
for $\phi = D(a)$, $a^\I \in D^\I$;
for $\phi = R(a,b)$, $\stru{a^\I,b^\I} \in R^\I$;
for $\phi = C \subs D$, $C^\I \subseteq D^\I$ (resp. for RIAs).
$\I$ is a \emph{model} of $\K$, denoted
$\I\dlmodels\K$, if it satisfies all axioms of $\K$.

Without loss of generality, we adopt the
{\em standard name assumption (SNA)} in the DL context 
(see~\cite{DBLP:journals/ai/EiterILST08,DBLP:journals/jacm/MotikR10} for more details).
That is, we assume an infinite subset
$\NI_S \subseteq\NI$ of individual constants, called {\em standard
  names} s.t. in every interpretation $\I$ we have (i)
 $\Delta^\I = \NI_S^\I = \{ c^\I \mid c \in \NI_S\}$; (ii) $c^\I
 \neq d^\I$, for every distinct $c,d \in \NI_S$. Thus, we may assume
 that $\Delta^I= \NI_S$ and $c^\I=c$ for each $c\in \NI_S$. 
  The \emph{unique name assumption (UNA)} 
	corresponds to assuming
  $c\neq d$ for all constants in $\NI\setminus \NI_S$ resp.\ occurring in 
  the knowledge base.
	
We confine here to 
knowledge bases without reflexivity axioms. The reason 
is that reflexivity
allows one to derive positive properties for any (named and unnamed) individual;
this complicates the treatment of defeasible axioms (cf.\ Discussion section).

\smallskip\noindent
\textbf{Datalog Programs and Answer Sets.}
We express our rules in
\emph{datalog with negation} 
under
answer sets semantics. 
In fact, we use here two kinds of 
negation\footnote{Strong negation can be easily emulated using weak
negation. While it does not yield higher expressiveness, it is
more convenient for presentation.}: 
strong (``classical'') negation $\non$ and weak \emph{(default) negation}\/ $\naf$
under the interpretation of answer sets semantics~\cite{gelf-lifs-91};
the latter is in particular needed for representing defeasibility.

A \emph{signature} is a tuple $\Pair{\DC}{\DP}$ of  a finite set $\DC$  of \emph{constants}
and a finite set $\DP$  of \emph{predicates}.
We assume a set $\DV$ of \emph{variables}; the elements of $\DC \cup
\DV$ are \emph{terms}.
%
An \emph{atom}
is of the form $p(t_1, \ldots, t_n)$
where $p \in \DP$ and $t_1$, \ldots, $t_n$, are terms.
A \emph{literal} $l$ is either a \emph{positive literal} $p$ or a 
\emph{negative literal} $\non p$, where  $p$ is an atom and $\non$ 
is 
strong negation. Literals of the form $p$, $\non p$ are \emph{complementary}.
We denote with $\neg. l$ the opposite of literal
$l$, i.e., $\neg.p = \non p$ and $\neg.\non p = p$ for an atom $p$.
A (datalog) rule $r$ is an expression: 
\begin{equation}
\label{rule}
a \leftarrow b_1, \dots, b_k, \naf b_{k+1}, \dots, \naf b_{m}.
\end{equation}
where $a, b_{1}, \dots, b_{m}$ are literals and $\naf$ is 
negation as failure (NAF).
We denote with $\Head(r)$ the head $a$ of rule $r$ and with
$\Body(r) = \{b_1, \dots, b_k,\naf b_{k+1}, \dots,$ $\naf b_{m}\}$ the body of $r$, respectively.
A (datalog) \emph{program} $P$ is a finite set of rules.
%
An atom (rule etc.) is \emph{ground}, if 
no variables occur in it. A \emph{ground substitution} $\sigma$ for $\Pair{\DC}{\DP}$
is any function $\sigma \,{:}\, \DV \to \DC$;
the \emph{ground instance} of an atom (rule, etc.) $\chi$ from
$\sigma$, denoted $\chi\sigma$, is obtained by replacing in $\chi$
each occurrence of variable $v \in \DV$ with $\sigma(v)$.
A \emph{fact} $H$ is a ground rule $r$ with empty body.
The \emph{grounding}\/ of a rule $r$, $\grd(r)$, is the set of all
ground instances of $r$, and the \emph{grounding}\/ of a program $P$
is $\grd(P) = \bigcup_{r\in P} \grd(r)$.

Given a program $P$, the \emph{(Herbrand) universe} $U_P$ of $P$ is the set of all
constants occurring in $P$ and the \emph{(Herbrand) base}
$B_P$ of $P$ is the set of all the 
ground literals 
constructable from the predicates in $P$ and the 
constants in $U_P$.
An \emph{interpretation} $I \subseteq B_P$ is any 
satisfiable subset of
$B_P$ (i.e., not containing complementary literals); 
a literal $l$ is \emph{true} in $I$, denoted $I\models l$,
if $l \in I$, and $l$ is \emph{false} in $I$ if $\neg.l$ is true.
%
Given a rule $r \in \grd(P)$,
we say that $\Body(r)$ is true in $I$, denoted $I\models \Body(r)$, if (i) $I\models b$ for each literal 
$b \in \Body(r)$ 
and (ii) $I\not\models b$ for each literal $\naf b\in \Body(r)$.
A rule r is \emph{satisfied} in $I$, denoted $I\models r$, if either 
$I\models \Head(r)$ or $I\not\models \Body(r)$.
An interpretation $I$ 
is a {\em model}\/ of $P$, denoted $I \models P$,
if $I\models r$ for each $r\in \grd(P)$;
moreover, $I$ is 
\emph{minimal}, 
if $I'\not\models P$ for each subset $I'\subset I$.

Given an interpretation $I$ for $P$, the (Gelfond-Lifschitz)
\emph{reduct} of $P$ w.r.t. $I$, denoted by $G_I (P)$, is the set of rules obtained from $\grd(P)$ by 
  (i) removing every rule $r$ such that 
 $I\models l$ for some $\naf l\in\Body(r)$; and
  (ii) removing the NAF part from the bodies of the remaining rules.
Then $I$ is an \emph{answer set} of $P$, if $I$ is a minimal
model of $G_I(P)$; the minimal model is
unique and exists iff $G_I(P)$ has some model. 
Moreover, if $M$ is an answer set for $P$, then $M$ is a minimal model of $P$.
We say that a literal $a \in B_P$ is a \emph{consequence} of $P$ and write
$P \models a$ if every answer set $M$ of $P$ fulfills $M \models a$.


\section{DL Knowledge Base with Justifiable Exceptions}
\label{sec:dkb}


In this paper we concentrate on reasoning on a DL knowledge base enriched
with \emph{defeasible axioms}, whose syntax and interpretation are 
analogous to~\cite{BozzatoES:18}. With respect to the contextual framework
presented in~\cite{BozzatoES:18}, this corresponds to reasoning
inside a single local context: while this simplifies 
presentation of the defeasibility aspects and the resulting reasoning method
for the case of $\DLliteR$, it can be generalized to the original case of multiple
local contexts.

\smallskip\noindent
\textbf{Syntax.}
%
Given a DL language $\Lcal_\Sigma$ based on a DL vocabulary 
$\Sigma = \NC_\Sigma \cup \NR_\Sigma \cup \NI_\Sigma$,
a \emph{defeasible axiom} is any expression of the form
$\default(\alpha)$, where $\alpha \in \Lcal_\Sigma$.

We denote with $\Lcal_\Sigma^\default$ the DL language extending
$\Lcal_\Sigma$ with the set of defeasible axioms in $\Lcal_\Sigma$.
On the base of such language,
we provide our definition of knowledge base with defeasible axioms.

\begin{definition}[defeasible knowledge base, DKB]
A \emph{defeasible knowledge base (DKB)} $\Kcal$ on
a vocabulary $\Sigma$ is a DL knowledge base over
$\Lcal^\default_\Sigma$.
\end{definition}
In the following, we tacitly consider DKBs based on $\DLliteR$.

\begin{example}
\label{ex:syntax}
We introduce a simple example showing the definition 
and interpretation of a defeasible existential axiom.
In the organization of a university research department, we want to specify that ``in general''
department members need also to teach at least a course. 
On the other hand, PhD students, while recognized as department members,
are not allowed to hold a course.
We can represent this scenario as a DKB $\K_{dept}$ where:

\begin{center}\footnotesize
  $\begin{array}{rl}
    \K_{dept}: & \left\{\begin{array}{l}
		           \default(\mi{DeptMember} \subs \exists \mi{hasCourse}),
							 \mi{Professor} \subs \mi{DeptMember},\\
							 \mi{PhDStudent} \subs \mi{DeptMember}, 
							 \mi{PhDStudent} \subs \non \exists \mi{hasCourse},\\
		           \mi{Professor}(\mi{alice}),\, \mi{PhDStudent}(\mi{bob})  
							\end{array}\right\}
  \end{array}$	
\end{center}
Intuitively, 
we want to override the
fact that there exists some course assigned to the PhD student $\mi{bob}$.
On the other hand, for the individual $\mi{alice}$ no overriding should happen
and the defeasible axiom can be applied.
\EndEx
\end{example}


\smallskip\noindent
\textbf{Semantics.}
We can now define a model based interpretation of DKBs,
in particular by providing a semantic characterization to 
defeasible axioms.

Similarly to the case of $\SROIQrl$ in~\cite{BozzatoES:18}, 
we can express $\DLliteR$ knowledge bases in first-order (FO)
logic, where every axiom $\alpha \in \Lcal_\Sigma$
is translated into an equivalent FO-sentence
$\forall\vec{x}.\phi_\alpha(\vec{x})$ where
$\vec{x}$ contains all free variables of $\phi_\alpha$ depending on
the type of the axiom.
The translation, depending on the axiom types, can be defined analogously to the
FO-translation presented in~\cite{BozzatoES:18}.
In the case of existential axioms of the kind
$\alpha = A \isa \exists R$, the FO-translation
$\phi_\alpha(\vec{x})$ is defined as:

\smallskip

\centerline{$A(x_1) \rightarrow R(x_1, f_\alpha(x_1))$\,;}

\smallskip

\noindent that is, we introduce a Skolem function $f_\alpha(x_1)$
which represents new ``existential'' individuals.
Formally, for every right existential axiom $\alpha \in \Lcal_\Sigma$,
we define a Skolem function 
$f_\alpha: \NI \mapsto \Ecal$
where $\Ecal$ is a set of new individual constants not appearing in $\NI$.
In particular, for a set of individual names $N \subseteq \NI$, we will write $sk(N)$ to denote the 
extension of $N$ with the set of Skolem constants for elements in $N$.

After this transformation
the resulting formulas $\phi_\alpha(\vec{x})$
amount semantically to Horn formulas, since
left-side concepts $C$ can be expressed by an existential
positive FO-formula, and right-side concepts $D$ by a conjunction
of Horn clauses. The following property from~\cite[Section 3.2]{BozzatoES:18} 
is then preserved for $\DLliteR$ knowledge bases.

\begin{lemma}
\label{lem:horn-equiv}
For a DL knowledge base $\K$ on $\Lcal_\Sigma$,
its FO-translation 
$\phi_\K \,{:=}\,\bigwedge_{\alpha \in \K}\!\!
  \forall\vec{x}\phi_\alpha(\vec{x})$
is semantically equivalent to a conjunction of universal Horn clauses.
\end{lemma}
With these considerations on the definition of FO-translation, 
we can now provide our definition of axiom instantiation: 

\begin{definition}[axiom instantiation]
Given an axiom $\alpha \in \Lcal_\Sigma$ with FO-translation
$\forall\vec{x}.\phi_\alpha(\vec{x})$, the instantiation  of $\alpha$
with a
tuple $\ee$ of individuals in $\NIs$, 
written $\alpha(\ee)$, is the
specialization of $\alpha$ to $\ee$, i.e., $\phi_\alpha(\ee)$,
depending on the type of $\alpha$.
\end{definition}
Note that, since we are assuming standard names, this basically means that
we can express instantiations (and exceptions) to any element of the domain (identified by a standard name in $\NI_\Sigma$).
We next introduce clashing assumptions and clashing sets.

\begin{definition}[clashing assumptions and sets]
A \emph{clashing assumption} is a pair $\stru{\alpha, \ee}$
such that $\alpha(\ee)$ is an axiom instantiation 
for an axiom $\alpha \in  \Lcal_\Sigma$. 
A \emph{clashing set} for a clashing assumption $\stru{\alpha,\ee}$
is a satisfiable set $S$ that consists of ABox assertions over
$\Lcal_\Sigma$ and negated ABox assertions of the forms $\neg
C(a)$ and $\neg R(a,b)$
such that
$S \cup \{\alpha(\ee)\}$ is unsatisfiable.
\end{definition}
A clashing assumption $\stru{\alpha, \ee}$
represents
that $\alpha(\ee)$ is not satisfiable, 
while a clashing set $S$ provides an assertional ``justification'' for the assumption of 
local overriding of $\alpha$ on~$\ee$.
We can then extend the notion of DL interpretation 
with a set of clashing assumptions. 

\begin{definition}[CAS-interpretation]
A \emph{CAS-interpretation} is a structure $\I_{\CAS} = \stru{\I, \casmap}$
where $\I = \stru{\Delta^\I, \cdot^\I}$ is a DL interpretation for $\Sigma$ and 
$\casmap$ is a set of clashing assumptions.
\end{definition}
By extending the notion of satisfaction with respect to CAS-interpretations,
we can disregard the application of defeasible axioms to the
exceptional elements in the sets of clashing assumptions. 
For convenience, we 
call two DL interpretations $\I_1$ and $\I_2$
\emph{$\NI$-congruent}, if
$c^{\I_1} = c^{\I_2}$ 
holds for every $c\in \NI$.

\begin{definition}[CAS-model]
\label{def:cas-model}
Given a DKB $\Kcal$, 
a CAS-interpretation $\I_{\CAS} = \stru{\I, \casmap}$
is a CAS-model for $\Kcal$ (denoted $\I_{\CAS} \models
  \Kcal$), if the following holds:
\begin{enumerate}[label=(\roman*)]
  \item
   for every $\alpha \in \Lcal_\Sigma$ in $\Kcal$, $\I \models \alpha$;
  \item
   for every  $\default(\alpha) \in \Kcal$ (where $\alpha \in \Lcal_\Sigma$),
   with $|\vec{x}|$-tuple $\vec{d}$ of elements 
   in $\NIs$ such that $\vec{d} \notin \{ \ee \mid \stru{\alpha,\ee} \in \casmap \}$, 
   we have $\I \models \phi_\alpha(\vec{d})$.
\end{enumerate}
\end{definition}
%
We say that a clashing assumption $\stru{\alpha, \ee} \in \casmap$ is
\emph{justified} for a $\CAS$ model $\I_{\CAS} = \stru{\I, \casmap}$,   
if some clashing set
$S = S_{\stru{\alpha,\ee}}$  exists such that, for every CAS-model
$\I_{\CAS}' = \stru{\I', \casmap}$ of $\Kcal$ 
that is $\NI$-congruent with $\I_{\CAS}$, 
it holds that $\I' \models S_{\stru{\alpha,\ee}}$.
We then consider as DKB models
only the CAS-models
where all clashing assumptions are justified.

\begin{definition}[justified CAS model and DKB model]
A $\CAS$ model $\I_{\CAS} = \stru{\I, \casmap}$ of 
a DKB $\Kcal$ is \emph{justified}, if every $\stru{\alpha, \ee} \in
\casmap$ is justified.
An interpretation $\I$ 
is a \emph{DKB model} of $\K$ (in
symbols, $\I\models\K$), if $\K$ has some  
justified $\CAS$ model $\I_{\CAS} = \stru{\I, \casmap}$.
\end{definition}

\begin{example}
	Reconsidering $\K_{dept}$ in Example~\ref{ex:syntax},
	a CAS-model providing the intended interpretation of 
	defeasible axioms is $\I_{\CAS_\mi{dept}} = \stru{\Ical, \chi_\mi{dept}}$ where
	  $\chi_\mi{dept} = \{\stru{\alpha, \mi{bob}}\}$
	with $\alpha = \mi{DeptMember} \subs \exists \mi{hasCourse}$.
	The fact that this model is justified is verifiable considering that for
	the clashing set $S = \{\mi{DeptMember}(\mi{bob}),$ $ \non \exists \mi{hasCourse}(\mi{bob})\}$
	we have $\I \models S$.
	On the other hand, note that a similar clashing assumption for $\mi{alice}$ is not justifiable:
	it is not possible from the contents of $\K_{dept}$ to derive a clashing set $S'$
	such that $S' \cup \{\alpha(\mi{alice})\}$ is unsatisfiable.
	By Definition~\ref{def:cas-model}, 
	this allows to apply $\alpha$ to this individual as expected and thus
	$\I \models \exists \mi{hasCourse}(\mi{alice})$.
	\EndEx
\end{example}
%
DKB-models have interesting properties similar as CKR-models in
\cite{BozzatoES:18}. In particular, we mention here that
for DKB-model $\I_{\CAS} = \stru{\I, \casmap}$, each clashing
assumption $\stru{\alpha,\ee}\in \casmap$  is over individuals of the knowledge
base, cf.\ \cite[Prop.~5, context focus]{BozzatoES:18}; this is because in absence of reflexivity, no positive properties 
(which occur in all clashing sets), can be proven for other
elements. Furthermore, 
the clashing assumptions are non-redundant, i.e., no NI-congruent DKB-model 
$\I'_{\CAS} = \stru{\I', \casmap'}$ exists such that $\casmap' \subset
\casmap$, cf.\ \cite[Prop.~6, minimality of justification]{BozzatoES:18}.


\section{Datalog Translation for $\DLliteR$ DKB}
\label{sec:translation}

We present a datalog translation for reasoning on $\DLliteR$
DKBs which refines the 
translation provided in~\cite{BozzatoES:18}.
The translation provides a reasoning method 
for positive instance queries w.r.t. entailment.
%
An important aspect of this translation is that, due to the form of $\DLliteR$ axioms,
no inference on disjunctive negative information 
is needed for the reasoning on derivations of clashing sets. 
Thus, differently from~\cite{BozzatoES:18}, 
reasoning by contradiction using 
``test environments'' is not needed and we can directly encode 
negative reasoning as rules on negative literals: with respect to the discussion in~\cite{BozzatoES:18},
we can say that $\DLliteR$ thus represents an inherently ``justification safe'' fragment
which then allows us to formulate such a direct datalog encoding.
With respect to the interpretation of right-hand side existential axioms, we follow the 
approach of~\cite{Krotzsch:10}: for every axiom
of the kind $\alpha = A \isa \exists R$, an auxiliary abstract individual $aux^\alpha$
is added in the translation to represent the class of all
$R$-successors introduced by $\alpha$.

%
We introduce a \emph{normal form} for axioms of $\DLliteR$
which allows us to simplify the formulation of reasoning rules:
the normal form axioms of $\DLliteR$ that we consider are shown in Table~\ref{tab:dlr-normalform}.
We can provide 
rules to transform any
$\DLliteR$ DKB into normal form and show that the rewritten DKB is equivalent
to the original.

\begin{table}[htp]%
\caption{Normal form for $\DKB$ axioms from $\Lcal_\Sigma$}
\label{tab:dlr-normalform}


\centerline{\small
$\begin{array}{c}
\hline\\[-1.75ex]
\multicolumn{1}{l}{\text{for $A,B,C \in \NCs$, $R \in \NRs$, $a,b \in \NIs$:}}\\[1ex]
\begin{array}{l}
A(a) \qquad  R(a,b) \qquad \non A(a) \qquad \non R(a,b) \qquad\	
A \subs B \qquad  A \subs \non C\\[1ex] 
  \exists R  \subs B   \qquad  A \subs \exists R \qquad\
  R  \subs T \qquad \mathrm{Dis}(R,S) \qquad \mathrm{Inv}(R,S) \qquad \mathrm{Irr}(R)\\ 
\end{array} 
\\[3ex] 
\hline
\end{array}$}
\end{table}


\smallskip\noindent
\textbf{Translation rules overview.}
We can now present the components of our datalog translation for 
$\DLliteR$ based DKBs.
%
As in the original formulation in~\cite{BozzatoES:14,BozzatoES:18},
which extended the encoding without defeasibility proposed in~\cite{BozzatoSerafini:13}
(inspired by the materialization calculus in~\cite{Krotzsch:10}),
the translation includes sets of \emph{input rules} (which encode DL axioms and signature in datalog),
\emph{deduction rules} (datalog rules providing instance level inference) and \emph{output rules}
(that encode in terms of a datalog fact the ABox assertion to be proved).
The translation is composed by the following sets of rules: 

\smallskip\noindent
\emph{$\DLliteR$ input and output rules:}
rules in $I_{dlr}$ encode as datalog facts the $\DLliteR$ axioms
and signature of the input DKB.
For example, in the case of existential axioms, 
these are translated as 
$A \subs \exists R \mapsto \{\supEx(A,R,aux^{\alpha})\}$:
note that this rule, in the spirit of~\cite{Krotzsch:10}, introduces an auxiliary element
$aux^\alpha$, which intuitively represents the class of
all new $R$-successors generated by the axiom $\alpha$.
Similarly, output rules in $O$ encode in datalog the
ABox assertions to be proved.
These rules are provided in Tables~\ref{tab:dlr-rules-tgl} and~\ref{tab:output-rules-tgl}.

\begin{table}[htp]%
\caption{$\DLliteR$ input and deduction rules}
\hrule\mbox{}\\
\textbf{$\DLliteR$ input translation $I_{dlr}(S)$}\\[.7ex]
\scalebox{.9}{
\small
$\begin{array}[t]{l@{\ \ }l}               
\mbox{(idlr-nom)} 
& a \in \NI \mapsto \{\nom(a)\}\\
\mbox{(idlr-cls)} 
& A \in \NC \mapsto \{\cls(A)\}\\
\mbox{(idlr-rol)} 
& R \in \NR \mapsto \{\rol(R)\}\\[1ex]

\mbox{(idlr-inst)} 
& A(a) \mapsto \{\insta(a,A)\} \\
\mbox{(idlr-inst2)} 
& \non A(a) \mapsto \{\non\insta(a,A)\} \\
\mbox{(idlr-triple)} 
& R(a,b) \mapsto \{\triplea(a,R,b)\} \\
\mbox{(irl-ntriple)} 
& \non R(a,b) \mapsto \{\non\triplea(a,R,b)\} \\
\end{array}$
\;
$\begin{array}[t]{l@{\ \ }l}               

\mbox{(idlr-subc)} 
& A \subs B \mapsto \{\subClass(A,B)\}\\
\mbox{(idlr-supnot)} 
& A \subs \non B \mapsto \{\supNot(A,B)\} \\
\mbox{(idlr-subex)} 
& \exists R \subs B \mapsto \{\subEx(R,B)\} \\

\mbox{(idlr-supex)} 
&  A \subs \exists R \mapsto \{\supEx(A,R,aux^{\alpha})\}\\[1ex]
        
\mbox{(idlr-subr)} 
& R \subs S \mapsto \{\subRole(R,S)\}\\
\mbox{(idlr-dis)} & \mathrm{Dis}(R,S)  \mapsto \{\pDis(R,S)\}\\
\mbox{(idlr-inv)} & \mathrm{Inv}(R,S) \mapsto \{\pInv(R,S)\}\\
\mbox{(idlr-irr)} & \mathrm{Irr}(R) \mapsto \{\pIrr(R)\}\\
\end{array}$}\\ 
\textbf{$\DLliteR$ deduction rules $P_{dlr}$}\\[.7ex]
\scalebox{.85}{
\small
$\begin{array}{l@{\;}r@{\ }r@{\ }l@{}}
	 \mbox{(pdlr-instd)} & \instd(x,z) & \rif & \insta(x,z).\\
	 \mbox{(pdlr-tripled)} & \tripled(x,r,y) & \rif & \triplea(x,r,y).\\[.5ex]

   
		
   \mbox{(pdlr-subc)} 
   &      \instd(x,z) & \rif & \subClass(y,z), \instd(x,y). \\
   \mbox{(pdlr-supnot}) 
   & \non\instd(x,z) & \rif & \supNot(y,z), \instd(x,y). \\
   \mbox{(pdlr-subex)} 
   & \instd(x,z) & \rif & \subEx(v,z), \tripled(x,v,x'). \\
   \mbox{(pdlr-supex)} 
   & \tripled(x,r,x') & \rif & \supEx(y,r,x'), \instd(x,y).\\[.5ex]


   \mbox{(pdlr-subr)} 
   & \tripled(x,w,x') & \rif & \subRole(v,w), \tripled(x,v,x'). \\
   \mbox{(pdlr-dis1)} & \non\tripled(x,u,y), & \rif & \pDis(u,v), \tripled(x,v,y).\\	
	 \mbox{(pdlr-dis2)} & \non\tripled(x,v,y), & \rif & \pDis(u,v), \tripled(x,u,y).\\	
   \mbox{(pdlr-inv1)} & \tripled(y,v,x) & \rif & \pInv(u,v), \tripled(x,u,y). \\
   \mbox{(pdlr-inv2)} & \tripled(y,u,x) & \rif & \pInv(u,v), \tripled(x,v,y). \\
   \mbox{(pdlr-irr)} & \non\tripled(x,u,x) & \rif & \pIrr(u), \const(x).\\[.5ex]

   \mbox{(pdlr-ninstd)} & \non\instd(x,z) & \rif & \non\insta(x,z).\\
   \mbox{(pdlr-ntripled)} & \non\tripled(x,r,y) & \rif & \non\triplea(x,r,y).\\[.5ex]

   \mbox{(pdlr-nsubc)} 
   &    \non\instd(x,y) & \rif & \subClass(y,z), \non\instd(x,z). \\
   
   \mbox{(pdlr-nsupnot}) 
   & \non\instd(x,y) & \rif & \supNot(y,z), \non\instd(x,z).\\
                          
   \mbox{(pdlr-nsubex)} 
   & \non\tripled(x,v,x') & \rif & \subEx(v,z), \const(x'), \non\instd(x,z).\\

   \mbox{(pdlr-nsupex)} 
	  & \non\instd(x,y) & \rif & \supEx(y,r,w), \const(x),\\ 
		                         &&& \AllNRel(x,r).\\[.5ex]	

   \mbox{(pdlr-nsubr)} 
   & \non\tripled(x,v,x') & \rif & \subRole(v,w), \non\tripled(x,w,x'). \\	 

   \mbox{(pdlr-ninv1)} & \non\tripled(y,v,x) & \rif & \pInv(u,v), \non\tripled(x,u,y). \\
   \mbox{(pdlr-ninv2)} & \non\tripled(y,u,x) & \rif & \pInv(u,v), \non\tripled(x,v,y). \\[1ex]	
	
		\mbox{(pdlr-allnrel1)} & \AllNRelStep(x,r,y) & \rif & \first(y), \non \tripled(x,r,y).\\
	  \mbox{(pdlr-allnrel2)} & \AllNRelStep(x,r,y) & \rif & \AllNRelStep(x,r,y'), \nextp(y',y),
		                                               \non \tripled(x,r,y).\\
	  \mbox{(pdlr-allnrel3)} & \AllNRel(x,r) & \rif & \lastp(y), \AllNRelStep(x,r,y).\\[1ex]	
		
  \end{array}$}\\[.5ex]
\hrule\mbox{}
\label{tab:dlr-rules-tgl}
\end{table}

\begin{table}[p]%
\caption{Output translation $O(\alpha)$}
\label{tab:output-rules-tgl} 
\hrule\mbox{}\\[1ex] 
\small\centering
$\begin{array}{l@{\ \ }l}
\mbox{(o-concept)} & A(a) \mapsto \{A(a)\} \\
\mbox{(o-role)} & R(a,b) \mapsto \{R(a,b)\} \\[1ex]
\end{array}$ 
\hrule
\end{table}

\smallskip\noindent
\emph{$\DLliteR$ deduction rules:}
rules in $P_{dlr}$ add deduction rules for ABox reasoning.
These rules are provided in Table~\ref{tab:dlr-rules-tgl}.
In the case of existential axioms,
the rule (pdlr-supex) introduces a new relation to the auxiliary individual as follows:
\begin{center}\small
   $\tripled(x,r,x')  \rif  \supEx(y,r,x'), \instd(x,y).$
\end{center}
In this translation the reasoning on negative information
is directly encoded by ``contrapositive'' versions of the rules.
For example, with respect to previous rule, we have:
\begin{center}\small
	$\non\instd(x,y) \rif \supEx(y,r,w), \const(x), \AllNRel(x,r).$
\end{center}
where $\AllNRel(x,r)$ verifies that $\non \triple(x,r,y)$ holds for all $\const(y)$
by an iteration over all constants.

\smallskip\noindent
\emph{Defeasible axioms input translations}: 
the set of input rules $I_\default$ 
(shown in Table~\ref{tab:input-default-tgl}) 
provides the translation of defeasible axioms $\default(\alpha)$ in the DKB: in other words,
they are used to specify that the axiom $\alpha$ need to be considered as
defeasible. 
For example, $\default(A \isa \exists R)$ is translated to 
$\defsupex(A, R, aux^{\alpha})$.
  
\begin{table}[p]%
\caption{Input rules $I_{\default}(S)$ for defeasible axioms}
\label{tab:input-default-tgl}


\hrule\mbox{}\\[1ex]
\small
\scalebox{.82}{
$\begin{array}{@{}l@{~}r@{~}l@{}}
 \mbox{(id-inst)} & \default(A(a))  & \mapsto \{\, \definst(A,a).\,\} \\
 \mbox{(id-triple)} & \default(R(a,b))  & \mapsto \{\, \deftriple(R,a,b).\,\} \\ 
 \mbox{(id-ninst)} & \default(\non A(a))  & \mapsto \{\, \defninst(A,a).\,\} \\
 \mbox{(id-ntriple)} & \default(\non R(a,b))  & \mapsto \{\, \defntriple(R,a,b).\,\} \\[1ex] 

 \mbox{(id-subc)} & \default(A \subs B)  & \mapsto \{\, \defsubs(A,B).\,\} \\  
 \mbox{(id-supnot)} & \default(A \subs \non B)  & \mapsto \{\, \defsupnot(A, B).\,\} \\[1ex]
\end{array}$
$\begin{array}{l@{~}r@{~}l}
 \mbox{(id-subex)} & \default(\exists R \subs B)  & \mapsto \{\, \defsubex(R, B).\,\} \\    
 \mbox{(id-supex)} & \default(A \subs \exists R)  & \mapsto \{\, \defsupex(A, R, aux^{\alpha}).\,\} \\[1ex]

 \mbox{(id-subr)} & \default(R \subs S)  & \mapsto \{\, \defsubr(R, S).\,\} \\    
 \mbox{(id-dis)} & \default(\mathrm{Dis}(R,S))  & \mapsto \{\, \defdis(R, S).\,\} \\   
 \mbox{(id-inv)} & \default(\mathrm{Inv}(R,S))  & \mapsto \{\, \definv(R,S).\,\} \\   
 \mbox{(id-irr)} & \default(\mathrm{Irr}(R))  & \mapsto \{\, \defirr(R).\,\}\\[1ex]
\end{array}$
} \hrule
\end{table}
	
\smallskip\noindent
\emph{Overriding rules:}
rules for defeasible axioms provide the different
conditions for the correct interpretation of defeasibility: the overriding rules
define conditions (corresponding to clashing sets) for 
recognizing an exceptional instance.
For example, for axioms of the form $\default(A \isa \exists R)$,
the translation introduces the rule:
\begin{center}\small
	$\ovr(\supEx,x,y,r,w) \rif \defsupex(y,r,w), \instd(x,y), \AllNRel(x,r).$
\end{center}
Note that in this version of the calculus, 
the reasoning on 
negative information (of the clashing sets) is directly encoded in the deduction
rules.
Overriding rules in $P_\default$, shown in Table~\ref{tab:ovr-rules-tgl}.

\begin{table}[p]%
\caption{Deduction rules $P_{\default}$ for defeasible axioms: overriding rules}
\label{tab:ovr-rules-tgl} 


\hrule\mbox{}\\[1ex]
\scalebox{.95}{
$\begin{array}{l@{\ \ }r@{\ \ }l}
 \mbox{(ovr-inst)} &
 \ovr(\insta,x,y)  \rif & \definst(x,y), \non \instd(x,y). \\
 \mbox{(ovr-triple)} &
 \ovr(\triplea,x,r,y) \rif & \deftriple(x,r,y), \non\tripled(x,r,y). \\[0.5ex]
 \mbox{(ovr-ninst)} &
 \ovr(\ninsta,x,y)  \rif & \defninst(x,y), \instd(x,y). \\
 \mbox{(ovr-ntriple)} &
 \ovr(\ntriplea,x,r,y) \rif & \defntriple(x,r,y), \tripled(x,r,y). \\[0.5ex]
 \mbox{(ovr-subc)} &
 \ovr(\subClass,x,y,z) \rif &
 \defsubs(y,z), \instd(x,y), \non \instd(x,z). \\
  \mbox{(ovr-supnot)} & 
  \ovr(\supNot,x,y,z) \rif &  \defsupnot(y,z), \instd(x,y), \instd(x,z).\\
													
  \mbox{(ovr-subex)} & 
  \ovr(\subEx,x,r,z) \rif &  
	\defsubex(r,z), \tripled(x,r,w), \non \instd(x,z).\\
  \mbox{(ovr-supex)} & 
  \ovr(\supEx,x,y,r,w) \rif & 
	\defsupex(y,r,w), \instd(x,y),\\
	&& \AllNRel(x,r).\\[0.5ex]

  \mbox{(ovr-subr)}  &
  \ovr(\subRole,x,y,r,s) \rif &
  \defsubr(r,s), \tripled(x,r,y),
  \non \tripled(x,s,y).\\
  \mbox{(ovr-dis)}  & 
  \ovr(\pDis,x,y,r,s) \rif &
  \defdis(r,s), \tripled(x,r,y), \tripled(x,s,y).\\
  \mbox{(ovr-inv1)}  & 
	\ovr(\pInv,x,y,r,s) \rif &
  \definv(r,s), \tripled(x,r,y),
	\non \tripled(y,s,x).\\	
  \mbox{(ovr-inv2)} & 
	\ovr(\pInv,x,y,r,s) \rif &
  \definv(r,s), \tripled(y,s,x), 
	\non \tripled(x,r,y).\\
  \mbox{(ovr-irr)}  & 
  \ovr(\pIrr,x,r,c) \rif &
  \defirr(r), \tripled(x,r,x).\\[1.5ex]
\end{array}$}
\hrule
\vspace{-2ex}
\end{table}

\smallskip\noindent
\emph{Defeasible application rules:}
another set of rules in $P_\default$ defines the defeasible application of such axioms:
intuitively, defeasible axioms are applied only to instances that have
not been recognized as exceptional.
For example, the rule (app-supex) applies a defeasible
existential axiom $\default(A \isa \exists R)$:
\begin{center}\small
  $\tripled(x,r,x') \rif \defsupex(y,r,x'), \instd(x,y),
                        \naf \ovr(\supEx,x,y,r,x').$
\end{center}
Defeasible application rules are provided in Table~\ref{tab:inheritance-rules-tgl}.

\begin{table}[p]%

\bigskip

\caption{Deduction rules $P_{\default}$ for defeasible axioms: application rules}
\label{tab:inheritance-rules-tgl} 

\hrule\mbox{}\\[1ex]
\scalebox{.9}{
\small
$\begin{array}{l@{\;}r@{\ }r@{\ }l@{}}
   \mbox{(app-inst)} 
   &   \instd(x,z) & \rif & \definst(x,z), \naf \ovr(\insta,x,z).\\
   \mbox{(app-triple)} 
   &   \tripled(x,r,y) & \rif & \deftriple(x,r,y), \naf \ovr(\triplea,x,r,y).\\[0.5ex]
   \mbox{(app-subc)} 
   &   \instd(x,z) & \rif 
	 & \defsubs(y,z), \instd(x,y), \naf \ovr(\subClass,x,y,z).\\
   \mbox{(app-supnot)} 
   &  \non\instd(x,z) & \rif 
	 & \defsupnot(y,z), \instd(x,y), \naf \ovr(\supNot,x,y,z).\\
													
   \mbox{(app-subex)} 
   & \instd(x,z) & \rif & \defsubex(v,z), \tripled(x,v,x'), \naf \ovr(\subEx,x,v,z). \\
   \mbox{(app-supex)}
   & \tripled(x,r,x') & \rif & \defsupex(y,r,x'), \instd(x,y),
                        \naf \ovr(\supEx,x,y,r,x'). \\[0.5ex]
													

   \mbox{(app-subr)} 
   & \tripled(x,w,x') & \rif & \defsubr(v,w), \tripled(x,v,x'), \naf \ovr(\subRole,x,y,v,w). \\
\mbox{(app-dis1)} 
   & \non\tripled(x,v,y) & \rif & \defdis(u,v), \tripled(x,u,y), \naf \ovr(\pDis,x,y,u,v). \\
\mbox{(app-dis2)} 
   & \non\tripled(x,u,y) & \rif & \defdis(u,v), \tripled(x,v,y), \naf \ovr(\pDis,x,y,u,v). \\
   \mbox{(app-inv1)} 
   & \tripled(y,v,x) & \rif & \definv(u,v), \tripled(x,u,y), \naf \ovr(\pInv,x,y,u,v).\\
   \mbox{(app-inv2)} 
   & \tripled(x,u,y) & \rif & \definv(u,v), \tripled(y,v,x), \naf \ovr(\pInv,x,y,u,v).\\       

 \mbox{(app-irr)} 
   & \non\tripled(x,u,x) & \rif & \defirr(u), \const(x) \naf \ovr(\pIrr,x,u). \\
		
   \mbox{(app-ninst)} 
   &   \non\instd(x,z) & \rif & \defninst(x,z), \naf \ovr(\ninsta,x,z).\\
   \mbox{(app-ntriple)} 
   &   \non\tripled(x,r,y) & \rif & \defntriple(x,r,y), \naf \ovr(\ntriplea,x,r,y).\\[0.5ex]

   \mbox{(app-nsubc)} 
   & \non\instd(x,y) & \rif 
	 & \defsubs(y,z), \non\instd(x,z), \naf \ovr(\subClass,x,y,z).\\

   \mbox{(app-nsupnot)} 
   &  \non\instd(x,y) & \rif 
	 & \defsupnot(y,z), \instd(x,z), \naf \ovr(\supNot,x,y,z).\\
													
   \mbox{(app-nsubex)} 
   & \non\tripled(x,v,x') & \rif & \defsubex(v,z), \const(x'), \non\instd(x,z), \naf \ovr(\subEx,x,v,z). \\
   \mbox{(app-nsupex)}
   & \non\instd(x,y) & \rif & \defsupex(y,r,x'), \const(x), \AllNRel(x,r),\\
                     &&&  \naf \ovr(\supEx,x,y,r,x').\\[0.5ex]												
												
   \mbox{(app-nsubr)} 
   & \non\tripled(x,v,y) & \rif & \defsubr(v,w), \non\tripled(x,w,y), \naf \ovr(\subRole,x,y,v,w). \\
   \mbox{(app-ninv1)} 
   & \non\tripled(y,v,x) & \rif & \definv(u,v), \non\tripled(x,u,y), \naf \ovr(\pInv,x,y,u,v).\\
   \mbox{(app-ninv2)} 
   & \non\tripled(x,u,y) & \rif & \definv(u,v), \non\tripled(y,v,x), \naf \ovr(\pInv,x,y,u,v).\\[1.5ex]	
	
	\end{array}$}
\hrule
\end{table}


\smallskip\noindent
\textbf{Translation process.}
Given a DKB $\K$ in $\DLliteR$ normal form, 
a program $PK(\K)$ that encodes query answering for $\K$ is obtained as:

\smallskip

\centerline{$PK(\K) = P_{dlr} \cup P_{\default} \cup 
     I_{dlr}(\K) \cup I_{\default}(\K)$}

\smallskip

\noindent 
Moreover, $PK(\K)$ is completed with a set of supporting facts
about constants:
for every literal $\nom(c)$, $\supEx(a,r,c)$ or $\defsupex(a,r,c)$ in $PK(\K)$,
$\const(c)$ is added to $PK(\K)$. Then, given an arbitrary enumeration $c_0, \dots, c_n$ 
s.t. each $\const(c_i) \in PK(\K)$, the facts $\first(c_0), \lastp(c_n)$
and $\nextp(c_i, c_{i+1})$ with 
$0 \leq i < n$ 
are added to $PK(\K)$.
Query answering $\K \models \alpha$ is then obtained 
by testing whether the (instance) query, translated to
datalog by $O(\alpha)$, is a consequence of $PK(\K)$, 
i.e., whether $PK(\K) \models O(\alpha)$ holds. 


\smallskip\noindent
\textbf{Correctness.}
The presented translation procedure provides
a sound and complete materialization calculus
for instance checking on $\DLliteR$ DKBs in normal form.

As in~\cite{BozzatoES:18}, 
the proof for this result can be verified by
establishing a correspondence between minimal justified models of $\K$
and answer sets of $PK(\K)$.
Besides the simpler structure of the final program,
the proof is simplified by the direct formulation of rules for 
negative reasoning.
Another new aspect of the proof in the case of $\DLliteR$
resides in the management of existential axioms,
since there is the need to define a correspondence between the auxiliary individuals
in the translation and the interpretation of existential axioms in the semantics:
we follow the approach of Kr\"{o}tzsch in~\cite{Krotzsch:10}, where,
in building the correspondence with justified models, auxiliary constants $aux^\alpha$ are mapped to
the class of Skolem individuals for existential axiom $\alpha$.
\comment{As in~\cite{BozzatoES:18}, in our translation
we consider UNA and \emph{named models}, i.e. interpretations restricted to $sk(N_\K)$, where $N_\K$
are the individuals that appear in the input $\K$.}
Thus we can show the correctness result on Herbrand models, that will be denoted $\hat{\I}(\casmap)$.

Let $\I_\CAS = \stru{\I, \casmap}$ be a justified named CAS-model. We define the set of overriding assumptions:
$$\OVR(\I_\CAS) = \{\, \ovr(p(\ee)) \;|\; \stru{\alpha, \ee} \in \casmap,\, I_{dlr}(\alpha) = p \,\}$$
Given a CAS-interpretation $\I_{\CAS}$, 
we can define a corresponding Herbrand interpretation 
$I(\I_{\CAS})$ for $PK(\K)$
by including the following atoms in it:
\begin{enumerate}[label=(\arabic*).]
\itemsep=0pt
\item 
all facts of $PK(\K)$;
\item 
 $\instd(a,A)$, if $\I \models A(a)$ and
 $\non\instd(a,A)$, if $\I \models \non A(a)$; 
\item 
  $\tripled(a,R,b)$, if  $\I \models R(a,b)$ and
  $\non\tripled(a,R,b)$, if  $\I \models \non R(a,b)$;	
\item 
  $\tripled(a,R,aux^\alpha)$, if $\I \models \exists R(a)$
	for $\alpha = A \isa \exists R$;
\item
	$\AllNRel(a,R)$
	if $\I \models \non \exists R(a)$;
\item
  each $\ovr$-literal from $\OVR(\I_{\CAS})$;
\end{enumerate}
%
%
The next proposition shows that the least Herbrand model of $\K$
can be represented by the answer sets of the program $PK(\K)$.

\begin{proposition}
\label{prop:correctness}
Let  $\K$ be a DKB in $\DLliteR$ normal form. Then:
%
	\begin{enumerate}[label=(\roman*).]
	\item 
	  for every (named) justified clashing assumption $\casmap$, 
		the interpretation $S = I(\hat{\I}(\casmap))$ is an answer set of $PK(\K)$;
	\item
	   every answer set $S$ of $PK(\K)$ is of the form 
		$S = I(\hat{\I}(\casmap))$ where $\casmap$ is a (named) justified clashing assumption for $\K$.
	\end{enumerate}
\end{proposition}
\begin{proof}[Sketch]
	We consider $S = I(\hat{\I}(\casmap))$ built as above and 
	reason over the reduct 
	$G_S(PK(\K))$ of $PK(\K)$ with respect $S$: 
	basically, $G_S(PK(\K))$ contains all ground rules from $PK(\K)$
	that are not falsified by some NAF literal in $S$,
	in particular it excludes application rules for the axiom instances
	that are recognized as overridden.
	
	Item (i) can then be proved by showing that given a
	justified $\chi$, $S$ is an answer set for 
	$G_S(PK(\K))$ (and thus $PK(\K)$):	
	the proof follows the same reasoning of the one in~\cite{BozzatoES:18},
	where 
	the fact that $I(\hat{\I}(\casmap))$ satisfies rules
	of the form (pdlr-supex) in $PK(\K)$
	is verified by the condition on existential formulas
	in the construction of the model above.
		
	For item (ii), we can show that from any answer set $S$
	we can build a justified model $\I_S$ for $\Kcal$
	such that $S = I(\hat{\I}(\casmap))$ holds.
	The model can be defined similarly to the original
	proof, but we need to consider auxiliary individuals in the
	domain of $\I_S$, that is thus defined as:
	$\Delta^{\I_S} = \{ c \;|\; c \in \NIs \}
	\cup \{aux^\alpha \;|\; \alpha = A \isa \exists R \in \K \}$.
	The result can then be proved by considering the 
	effect of deduction rules for existential axioms in $G_S(PK(\K))$:
	auxiliary individuals provide the domain elements in $\I_S$
	needed to verify this kind of axioms.
	The justification of the model
	follows by verifying that the new formulation of 
	overriding rules correctly encode the 
	possible clashing sets for the input defeasible axioms.
\end{proof}

\noindent
The correctness result directly follows from 
Proposition~\ref{prop:correctness}.
\begin{theorem}
\label{thm:encode-c-entailment}
Let  $\K$ be a DKB in $\DLliteR$ normal form, and let 
$\alpha \in \Lcal_\Sigma$ such that $O(\alpha)$ is defined. 
Then $\K \models \alpha$ iff $PK(\K) \models O(\alpha)$.
\end{theorem}
%
We note that by further normalization of the DKB, the translation
can be slimmed at the cost of new symbols. E.g., existential
restrictions $\exists R$ can be named ($A_{\exists R} \equiv \exists R$)
and replaced throughout by $A_{\exists R}$; however, we refrain here from further discussion.


\section{Complexity of Reasoning Problems}
\label{sec:complexity}

We first consider the satisfiability problem, i.e., deciding whether a
given $\DLliteR$ DKB has some DKB-model. As it turns out, defeasible
axioms do not increase the complexity with respect to satisfiability
of $\DLliteR$, due to the following property. Let $ind(\K)$ denote
the set of individuals occurring in $\K$.

\begin{proposition}
\label{prop:DKB-existence}
Let $\DKB$ be a $\DLliteR$ DKB, and let $\casmap_0 = \{
\stru{\alpha,\ee} \mid D(\alpha) \in \DKB$, $\ee$ is over $ind(\DKB)$ \}
be the clashing assumption that makes an exception to every
defeasible axiom over the individuals occurring in $\DKB$.
Then $\DKB$ has some DKB-model iff
$\DKB$ has some CAS-model $\I_{\CAS} = \stru{\I, \casmap_0}$.
\end{proposition}
Informally, the only if direction holds because any DKB-model of
$\DKB$ is also a CAS-model of $\DKB$; as justified exceptions are
only on $ind(\DKB)$, and making more exceptions does not destroy
CAS-modelhood, some CAS-model with clashing assumptions $\casmap_0$
exists. Conversely, if $\DKB$ has some CAS-model of the form
$\I_{\CAS} = \stru{\I, \casmap_0}$, a justified CAS-model can be
obtained by setting $\casmap = \casmap_0$ and trying to remove, one by
one, each clashing assumption $\stru{\alpha,\ee}$ from $\casmap$; this
is possible, if $\DKB$ has some NI-congruent model
$\stru{\I',\casmap\setminus\{ \stru{\alpha,\ee}\}}$. After looping
through all clashing assumptions in $\casmap_0$, we have that some 
some NI-congruent model $\stru{\I', \casmap}$ exists that is
justified. 

Thus, DKB-satisfiability testing boils down to CAS-satisfiability
checking, which can be done using the datalog encoding described in
the previous section. From the particular form of that encoding, we
obtain the following result.

\begin{theorem}
\label{theo:DKB-sat}
Deciding whether a given $\DLliteR$ DKB $\DKB$ has some DKB-model is
$\nlogspace$-complete in combined complexity and FO-rewritable 
in data complexity.
\end{theorem}
To see this, the program $PK(\DKB)$ for $\DKB$ has in each rule at most one
literal with an intentional predicate in the body, i.e., a predicate
that is defined by proper rules. Thus, we have a linear datalog
program with bounded predicate arity, for which derivability of an atom
is feasible in nondeterministic logspace, as this can be reduced to 
a graph reachability problem in logarithmic space. The \nlogspace-hardness is inherited from the combined complexity of KB satisfiability
$\DLliteR$, which is $\nlogspace$-complete.

As regards data-complexity, it is
well-known that instance checking and similarly satisfiability testing for $\DLliteR$
are FO-rewritable \cite{CalvaneseGLLR07};
this has been shown by a reformulation algorithm, which informally
unfolds the axioms $\alpha(\vec{x})$ (i.e., performs resolution
viewing axioms as clauses), such that deriving an instance $A(a)$
reduces to presence of certain assertions in the ABox. This unfolding
can be adorned by typing each argument $x\,{\in}\,\vec{x}$ of an axiom to whether it is
an individual from the DKB  (type i), or an unnamed individual (type u); for
example, $\alpha(x) = A \isa B$ yields $\alpha_{\rm i}(x)$ and
$\alpha_{\rm u}(x)$. The typing carries over to unfolded
axioms.  In unfolding, one omits typed versions of defeasible axioms
$D(\alpha(\vec{x}))$, which w.l.o.g.\
have no existential restrictions; e.g., for $D(\alpha(x)) = D(B \isa
C)$, one omits $\alpha_{\rm i}(x)$. In this way, instance derivation 
(and similarly satisfiability testing) is reduced to
presence of certain ABox assertions again.

On the other hand, entailment checking is intractable: while some
justified model is constructible in polynomial time, there can be 
exponentially many clashing assumptions for such models, even under
UNA; finding a DKB model that violates an axiom turns out to be difficult.

\begin{theorem}
\label{theo:DKB-entail-conp}
Given a DKB $\DKB$  and an axiom $\alpha$,  deciding
$\DKB\models\alpha$ is \conp-complete; this holds also for data
complexity and instance checking, i.e., $\alpha$ is of the form $A(a)$
for some assertion $A(a)$.
\end{theorem}

\begin{proof}[Sketch]
In order to refute $\DKB\models \alpha$,
\comment{we can exhibit that}
a justified CAS-model  $\IC_{\CAS} = \stru{\I,\casmap}$ of $\DKB$ 
named relative to $sk(N)$ exists
such that $\I \not\models \alpha$, 
with $N_\DKB \subseteq N \subseteq \NI \setminus \NI_S$
\comment{and where $N$ is of small (linear) size and
includes fresh individual names such that 
$\I$ violates the instance of $\alpha$ for some elements $\ee$ over $sk(N)$.}
We can guess clashing assumptions $\casmap$ over $N$, 
where each $\stru{\alpha,\ee} \in \casmap$ has a unique clashing set
$S_{\alpha(\ee)}$, and a partial interpretation
over $N$, and check derivability of all $S_{\alpha(\ee)}$ and
that the interpretation extends to a model of $\DKB$ relative to
$sk(N)$ in polynomial time. Thus, we overall obtain
membership of entailment in \conp. 

The \conp-hardness can be shown by a reduction from 
inconsistency-tolerant reasoning from $\DLliteR$
KBs under AR-semantics \cite{DBLP:conf/rr/LemboLRRS10}. Given 
a $\DLliteR$ KB $\Kcal = \Acal\cup \Tcal$ with ABox $\Acal$ and TBox
$\Tcal$,
a repair is a maximal subset $\Acal'\subseteq \Acal$ such
that $\Kcal' = \Acal'\cup\Tcal$ is satisfiable; an assertion $\alpha$
is AR-entailed by $\Kcal$, if $\Kcal'\models \alpha$ for every repair
$\Kcal'$ of $\Kcal$. As shown by Lembo et al.{},
deciding AR-entailment is \conp-hard; this continues to hold under UNA and if all
assertions involve only concept resp.\ role names.

Let $\hat{\DKB} \,{=}\, \Tcal \,{\cup}\, \{ D(\alpha) \mid \alpha \,{\in}\, \Acal\}$, i.e., 
all assertions from $\Kcal$ are defeasible. 
As easily seen, the
maximal repairs $\Acal'$ correspond to the justified clashing assumptions
by $\casmap = \{ \stru{\alpha,\ee}\mid$ $ \alpha(\ee) \in
\Acal\setminus\Acal'\}$. Thus, $\Kcal$ AR-entails $\alpha$ iff
$\hat{\DKB} \models \alpha$,
proving \conp-hardness. 

To show the result for data complexity, if we do not allow for
defeasible assertions, we can adjust the transformation, where
we emulate $D(A(a))$ by an axiom $D(A' \isa A)$ and make the assertion
$A'(a)$, where $A'$ is a fresh concept name; similarly $D(R(a,b))$ is
emulated by $D(R' \isa R)$ plus $R'(a,b)$, where $R'$ is a fresh role
name. As Lembo et al.{} proved \conp-hardness 
under data-complexity, the claimed result follows.
\end{proof}
\noindent
We observe that the \conp-hardness proof in \cite{DBLP:conf/rr/LemboLRRS10} 
used many role restrictions and inverse roles; for 
combined complexity, \conp-hardness of entailment 
in absence of any role names can be derived from  results about propositional circumscription in 
\cite{DBLP:journals/jcss/CadoliL94}. In particular,
\cite[Theorem 16]{DBLP:journals/jcss/CadoliL94} showed that deciding whether 
an atom $z$ is a circumscriptive consequence of 
a positive propositional 2CNF $F$ if all variables except $z$ are
minimized (i.e., in circumscription notation $\mathit{CIRC}(F;P,\emptyset;\{z\})\models z$),
is \conp-hard;%
\footnote{The models of $\mathit{CIRC}(F;P,\emptyset;\{z\})$ are all
  models $M$ of $F$ such that no model $M'$ of $F$ exists with
  $M'\setminus \{z\} \subset M\setminus \{z\}$.} such an inference can
be easily emulated by entailment from a DKB constructed from $F$ and
$z$, where propositional variables are used as concept names.

Indeed, for each clause $c = x \lor y$ in $F$, we add to $\DKB$ an
axiom $x \isa \non y$ if $z\neq x,y$ and an axiom $x \isa z$
(resp.\ $y\isa z$) if $z\,{=}\,y$ (resp.\ $x\,{=}\,z$).  Furthermore, for each
variable $x\neq z$, we add 
$D(x(a))$, where $a$
is a fixed individual. This effects that justified DKB-models of
$\DKB$ correspond to the models of
$\mathit{CIRC}(F;P,\emptyset;\{z\})$, where the minimality of exceptions
in justified DKB-models emulates the minimality of circumscription
models; thus, $\DKB \models  z(a)$ iff
$\mathit{CIRC}(F;P,\emptyset;\{z\})\models z$. Similarly as
above, defeasible assertions could be moved to 
defeasible axioms $D(c \isa v)$ with a single assertion $c(a)$.

While this establishes \conp-hardness of entailment for combined
complexity under UNA when roles are absent, the data complexity is
tractable; this is because we can
consider the axioms for individuals $a$ separately, and 
if the GCI axioms are fixed only few
axioms per individual exist. This is similar if role axioms but no
existential restrictions are
permitted, as we can concentrate on
the pairs $a,b$ and $b,a$ of individuals.
The questions remains how much of the latter is possible while staying tractable.

\section{Discussion and Conclusion}

\noindent
\textbf{Related works.}
The relation of the justified exception approach to nonmonotonic
description logics was discussed in \cite{BozzatoES:18}, where in
particular an in-depth comparison w.r.t.\
typicality in DLs \cite{GiordanoGOP:13},
normality \cite{DBLP:journals/jair/BonattiFS11} and
overriding \cite{BonattiFPS:15} was given. 
A distinctive feature of our approach, linked to the 
interpretation of exception candidates as different clashing assumptions,
is the possibility to ``reason by cases'' inside the alternative
justified models (cf.\ the discussion of the classic Nixon Diamond example
\cite[Section 7.4]{BozzatoES:18}).
%
%
%
The introduction of non-monotonic features in the $\DLlite$ family
and, more in general, to low complexity DLs has been 
the subject of many works, mostly with the goal of 
preserving the low complexity properties of the base logic in the extension.
For example, in~\cite{DBLP:journals/jair/BonattiFS11} a study of the 
complexity of reasoning with circumscription in $\DLliteR$ and
$\cal{EL}$ was presented.
Similarly, in~\cite{GiordanoGOP:11} the authors studied the application of their 
typicality approach to $\DLlite_c$ and $\cal{EL}^\bot$.
A recent work in this direction
is \cite{DBLP:journals/ijar/PenselT18}, where a defeasible version of $\cal{EL}^\bot$
was obtained by modelling higher typicality
by extending classical canonical
models in $\cal{EL}^\bot$ with 
multiple representatives of concepts and individuals.

\smallskip\noindent
\textbf{Summary and future directions.}
In this paper, we applied the justified exception approach from~\cite{BozzatoES:18} 
to reason on $\DLliteR$ knowledge bases with defeasible axioms.
We have shown that the limited language of $\DLliteR$ allows us to formulate a
direct datalog translation to reason on derivations for negative information
in instance checking.

The form of $\DLliteR$ axioms enables us to concentrate 
on exceptions in absence of reflexivity over the
individuals known from the KB: however, we are interested in studying
the case of languages allowing exceptions on unnamed individuals
(generated by existential axioms) by providing them with a suitable
semantic characterization. In particular, if reflexivity axioms are
allowed, positive properties are provable for unnamed individuals
(i.e., standard names). To account for this, multiple auxiliary
elements $aux^\alpha$ may be necessary to enable different exceptions
for unnamed individuals reached from different individuals; this
remains for further investigation.

Moreover, we plan to apply the current results on $\DLliteR$ in the framework of Contextualized
Knowledge Repositories with hierarchies as in~\cite{DBLP:conf/kr/BozzatoSE18}.

\nocite{BozzatoES:18}
\bibliographystyle{splncs04}
\bibliography{bibliography}


\end{document}